\documentclass[journal,draftclsnofoot,onecolumn,12pt]{IEEEtran}
\usepackage{cite}
\usepackage{amsmath,amssymb,amsfonts}
\usepackage{algorithmic}
\usepackage{graphicx}
\usepackage{textcomp}
\usepackage{xcolor}
\usepackage{bm}
\usepackage{mathrsfs}
\usepackage{tikz}
\usepackage{url}
\usepackage{verbatim}
\usepackage{amsthm}
\usepackage{tabularx} 
\usepackage{arydshln}
\usepackage{tikz}
\usepackage{array}
\usepackage{nicematrix}

 \ExplSyntaxOn
\prg_new_conditional:Npnn \int_if_zero:n #1 { p , T , F , TF }
{
	\if_int_compare:w \int_eval:w #1 = \c_zero_int
	\prg_return_true:
	\else:
	\prg_return_false:
	\fi:
}
\ExplSyntaxOff

\newtheorem{lemma}{Lemma}
\newtheorem{theorem}{Theorem}
\newtheorem{remark}{Remark}
 
\newtheorem{proposition}{Proposition}
\newtheorem{assumption}{Assumption}
\theoremstyle{definition}
\newcounter{example}
\newenvironment{example}[1][]{\refstepcounter{example}\par\medskip
	\noindent \textbf{Example~\arabic{example}. #1} \rmfamily}{\medskip}

\def\BibTeX{{\rm B\kern-.05em{\sc i\kern-.025em b}\kern-.08em
    T\kern-.1667em\lower.7ex\hbox{E}\kern-.125emX}}
\begin{document}

\title{A Short Proof of Coding Theorems for Reed-Muller Codes Under a Mild Assumption}

\author{Xiao Ma,~\IEEEmembership{Member,~IEEE}
	\thanks{The author is with the School of Computer Science and Engineering, and also with  the Guangdong Key Laboratory of Information Security Technology, Sun Yat-sen University, Guangzhou 510006, China (e-mail: maxiao@mail.sysu.edu.cn).}}

\markboth{Journal of \LaTeX\ Class Files,~Vol.~1, No.~2, March~2025}%
{Shell \MakeLowercase{\rmit{et al.}}: A Sample Article Using IEEEtran.cls for IEEE Journals}

\maketitle

\begin{abstract}
	
In this paper, by treating Reed-Muller~(RM) codes as a special class of low-density parity-check~(LDPC) codes and assuming that sub-blocks of the parity-check matrix are randomly interleaved to each other as Gallager's codes, we present a short proof that RM codes are entropy-achieving as source coding for Bernoulli sources and capacity-achieving as channel coding for binary memoryless symmetric~(BMS) channels, also known as memoryless binary-input output-symmetric~(BIOS) channels, in terms of bit error rate~(BER) under maximum-likelihood~(ML) decoding.
\end{abstract}


\section{Introduction}
Reed-Muller~(RM) codes are named after Reed and Muller~\cite{Muller1954, Reed1954},  which were introduced in 1954 and are one of the oldest families of codes. Although a sequence of RM codes with a fixed rate has a vanishing relative distance ratio, it was believed, in Shannon's setting~\cite{shannon1948mathematical}, that RM codes are good enough for correcting random errors. 
In 1993, Shu Lin gave a talk, entitled  ``RM Codes are Not So Bad"~\cite{lin1993}, and then gave the transparencies to Forney\footnote{Personal communication with Prof. Shu Lin, Mar., 2025.}. 
In an invited paper~\cite{costello2007channel}, Costello and Forney in 2007 wrote that ``Indeed, with optimum decoding, RM codes may be good enough to reach the Shannon limit on the AWGN channel". Shortly after the polar code breakthrough~\cite{ref1Arikan},  there is a growing conviction that RM codes can achieve capacity on any binary memoryless symmetric~(BMS) channels, given the close relationship between polar codes and RM codes and that the RM codes typically have larger minimum Hamming distance. Arıkan pointed out in a survey~\cite{arikan2010survey} that whether RM codes are capacity-achieving under maximum-likelihood~(ML) decoding is a fundamental open problem in coding theory. 
Most recently, Reeves and Pfister~\cite{reeves2023reed} proved that RM codes achieve the capacity of the BMS channels in terms of bit error rate~(BER), while Abbe and Sandon~\cite{Abbe2023} proved that RM codes are also capacity-achieving on symmetric channels in terms of block error rate~(BLER).

In this paper, we show that RM codes can be viewed as a class of low-density parity-check~(LDPC) codes, resembling Gallager's construction~\cite{gallager1962low}. Then, by assuming that different blocks are randomly interleaved to each other, we present a short proof for coding theorems of RM codes, in terms of BER under ML decoding. 



\section{System Model and Main Results}

\subsection{System Model}
\begin{figure}[!t]
	\centering
	\includegraphics[width=3.4in]{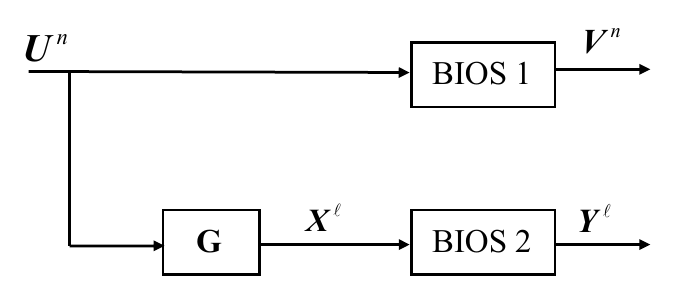}
	\caption{System model.}
	\label{system_model}
\end{figure}

Recently, a new framework to prove coding theorems for linear codes has been proposed~\cite{WYX2021}, as depicted in Fig.~\ref{system_model}, where~$\bm U^{n}\in \mathbb{F}_2^{n}$ is a segment of a Bernoulli process with a success probability of~$\theta = P_U(1) \leq 1/2$ and referred to as the message bits to be transmitted, ${\mathbf{G}}$ is a  binary matrix of size $n\times \ell$ and~$\bm X^{\ell}=\bm U^n{\mathbf{G}}\in \mathbb{F}_{2}^{\ell}$ is referred to as parity-check bits corresponding to $\bm U^n$. The message bits $\bm U^n$ and the parity-check bits $\bm X^{\ell}$ are transmitted through two~(possibly different) binary-input output-symmetric~(BIOS) channels, resulting in~$\bm V^n$ and~$\bm Y^{\ell}$, respectively. A  BIOS channel is characterized by an input~$x \in \mathcal{X}=\mathbb{F}_{2}$, an output set~$\mathcal{Y}$~(discrete or continuous), and a conditional probability mass (or density) function $\{P_{Y|X}(y|x)\big| x\in\mathbb{F}_2, y\in\mathcal{Y}\}$ which satisfies the symmetric condition that $P_{Y|X}(y|1)=P_{Y|X}(\pi(y)|0)$ for some bijective mapping $\pi : \mathcal{Y}\rightarrow\mathcal{Y}$ with $\pi^{-1}(\pi(y))=y$. For simplicity, we assume that the BIOS channels are memoryless, meaning that~$P_{\bm{V}|\bm{U}}(\bm u|\bm v)=\prod\limits_{i=0}^{n-1}P_{V|U}(v_i|u_i)$ and~$P_{\bm{Y}|\bm{X}}(\bm y|\bm x)=\prod\limits_{i=0}^{\ell-1}P_{Y|X}(y_i|x_i)$. This model is also known as a BMS channel. The entropy of a Bernoulli source is defined as $H(U) = -\theta\log_2(\theta) -(1-\theta)\log_2(1-\theta)$, while the channel capacity of a BIOS channel is given by $C = I(X; Y)$ with $I(X; Y)$ being the mutual information and $X$ being a uniform binary random variable with $P_X(0) = P_X(1) = 1/2$.

The task of the receiver is to recover $\bm U^n$  from $(\bm V^n,\bm Y^{\ell})$. This can be done by a two-step decoding algorithm as follows. First, based on $\bm V^n$, the receiver finds a list of size at most $2^{n(H(U|V)+\epsilon)}$ which contains $\bm U^n$ with high probability. Second, based on $\bm Y^{\ell}$, the receiver selects $\hat{\bm U}^n$ from the list as the estimate such that $\hat{\bm X}^{\ell} = \hat{\bm U}^n \mathbf{G}$ is the most likely one. This two-step decoding reveals a new mechanism pertained to systematic linear codes, where the information part and the parity-check part can play different roles. Specifically, the noisy information part helps to limit the list size of the candidate codewords, while the noisy~(or noiseless) parity-check part is used to identify the ML codeword from the list. The BLER is defined as ${\rm Pr}\{\hat{\bm U}^n\neq \bm U^n\}$, while the BER is defined simply as $\mathbb{E}[W_H(\hat{\bm U}^n -\bm U^n)]/ n$ with $W_H(\cdot)$ denoting the Hamming weight. 

We are primarily concerned with the following question: for sufficiently large $n$, how many parity-check bits are required for reliably transmitting $\bm U^n$? Intuitively, the reliable transmission of $\bm U^n$ can be achieved if $nH(U|V ) < \ell I(X; Y )$. The limit of the ratio $\ell/n$ as $n$ goes to infinity is considered in this paper, which represents the average number of parity-check bits required per message bit for reliable transmission. The framework unifies the following problems in coding theory.

\begin{itemize}
	\item If BIOS 1 and BIOS 2 are the same, and $U$ is uniform~(i.e., $\theta = 1/2$), the problem is equivalent to the channel coding problem. From $nH(U|V ) < \ell I(X; Y)$, we have $\ell/n > H(U|V )/I(X; Y)$, which is equivalent to $n/(n + \ell) < C$ and $n/(n + \ell)$ is the code rate for channel coding. This can be verified by noting that $C = 1 - H(X|Y)$ and $H(U|V ) = H(X|Y )$.
	\item If BIOS 1 is a totally erased channel~(meaning that $\bm U^n$ is not transmitted) and BIOS 2 is noiseless, the problem is equivalent to the source coding problem. From $nH(U|V ) < \ell I(X; Y )$, we have $\ell/n > H(U)$ and $\ell/n$ is the code rate for the source coding. This can be verified by noting that $H(U|V ) = H(U)$ and $I(X; Y ) = 1$.
	\item If BIOS 1 is a totally erased channel, and BIOS 2 has noise, the problem can be viewed as the joint source channel coding~(JSCC) problem for Bernoulli sources. From $nH(U|V ) <
	\ell I(X; Y )$, we have $\ell/n > H(U)/I(X; Y )$, where $\ell/n$ is the transmission ratio and the
	limit of the minimum transmission ratio~(LMTR) is $H(U)/C$~\cite{Csiszar2011}.
	\item If the source is uniform~(i.e., $\theta = 1/2$) and BIOS 1 is a noisy channel while BIOS 2 is a noiseless channel, the problem	is equivalent to transmit the uniform sources with the parity-check codes specified by the parity-check matrix $\mathbf{H} = \mathbf{G}^{T}$. From $nH(U|V ) < \ell I(X; Y )$, we have $\ell/n > H(U|V )$ for $I(X; Y ) = 1$. Equivalently, we have $(n - \ell)/n < C$, where $(n - \ell)/n$ is the design code rate of the parity-check code.
\end{itemize}

\subsection{Reed-Muller Codes}
A Reed-Muller code of length $n=2^m$ and order $r$, denoted by RM$(r,m)$, has dimension $\ell=\sum_{i=0}^{r} \binom{m}{i}$ and minimum distance $d_{\text{min}} = 2^{m-r}$. A codeword can be specified by evaluating a Boolean function $f(\bm v) = f(v_1,v_2,\cdots, v_m)$ of degree not greater than $r$ with $\bm v = (v_1,v_2,\cdots, v_m)$ ranging over $\mathbb{F}_2^m$. The RM$(r,m)$ code can be generated by all monomials of degree not greater than $r$.  We have the following proposition. 
\begin{proposition}\label{mrlimit}
	Let $R\in(0,1)$ be fixed. For a positive integer $m$, let $n = 2^m$ and $r$ be the minimum integer such that $\sum_{i=0}^{r} \binom{m}{i} \geq nR$. Then $\lim\limits_{m\rightarrow \infty}(m-r)=\infty$ and hence $\lim\limits_{m\rightarrow \infty}d_{\textrm{min}}=\infty$.
\end{proposition}
\begin{proof}
	This can be proved by contradiction. If this is not the case, there exists a subsequence $(m,r)$ such that $m-r<t$ for some fixed $t$. Then we have 
	\begin{equation}
		\begin{aligned}
			R &>  \frac{\sum_{i=0}^{r-1} \binom{m}{i} }{2^m} = \frac{2^m-\sum_{i=r}^{m} \binom{m}{i} }{2^m} \\
			&> 1-\frac{ t\binom{m}{t} }{2^m} \overset{(*)}{\geq} 1-\frac{ t \cdot 2^{mH(t/m)}}{2^m}  \rightarrow 1 \quad \text{as } m\rightarrow \infty,
		\end{aligned}
	\end{equation}
	where the inequality $(*)$ follows from the fact that $\binom{m}{t} \leq 2^{mH(t/m)}$ with $H(\cdot)$ being the entropy function~\cite[Example 11.1.3]{cover1999elements}. This contradicts the assumption that $R < 1$, completing the proof.
\end{proof}

The RM$(r,m)$ code can also be generated by all  disjunctive normal forms of degree not greater than $r$. Here, a disjunctive normal form of degree $r$ is defined as $f = \prod\limits_{i=1}^{r}f_i$ where $f_i \in \{v_1,v_2,\cdots,v_m\} \cup \{1+v_1,1+v_2,\cdots, 1+v_m\}$ and $\{f_i, 1+f_i\} \neq\{f_j, 1+f_j\} $ for $i \neq j$. We have the following proposition~\cite{macwilliams1977theory}.

\begin{proposition}
	The RM$(r,m)$ code can be generated from the disjunctive normal form of degree~(exactly) $r$. 	
\end{proposition} 

\begin{proof}
	It suffices to prove that any monomial of degree less than $r$ can be represented as a sum of disjunctive normal forms of degree $r$. Let $f$ be a monomial of degree $t<r$. Without loss of generality, we assume that $f$ does not have factors $v_i$, $1\leq i \leq r-t$. Then we write $f = f\cdot \prod\limits_{i=1}^{r-t}(v_i+1+v_i)$, which can be expanded as a sum of disjunctive normal form of degree $r$. 
\end{proof}	
%
%

We also know that a disjunctive normal form of degree $r$ corresponds to a codeword of weight $2^{m-r}$, which goes to infinity as $n$ goes to infinity but $\ell/n$ fixes. In this paper, we treat a Boolean function as a codeword in the column form. 

Consider the RM$(r,m)$ code with $\ell = \sum\limits_{i=0}^{r}\binom{m}{i}$. Let $\mathcal{S}_1, \mathcal{S}_2, \cdots, \mathcal{S}_t$ be $t = \binom{m}{r}$ distinct subsets of $\{v_1,v_2,\cdots, v_m\}$ and $|\mathcal{S}_i| = r$, $1\leq i \leq t$. To be specific, we may arrange $\mathcal{S}_i$ in lexicographic order. For $1\leq i \leq t$, let $\mathbf{G}_i$ be a matrix of size $n\times 2^r$ with columns consisting of disjunctive normal forms of degree $r$ defined on $\mathcal{S}_i$.
It can be verified that the columns within a given sub-matrix $\mathbf{G}_i$ are non-overlapped~(meaning that any two columns do not have common non-zero coordinates) and have column weight $2^{m-r}$. Obviously, the columns in $\mathbf{G}_i$ are linearly independent, while the columns across different sub-matrices can be linearly dependent. Also notice that the matrix $\mathbf{G}_i$ with $i > 1$ is an interleaved version of $\mathbf{G}_1$, resembling Gallager's $\left(2^{m-r}, \binom{m}{r}\right)$-regular LDPC codes~\cite{gallager1962low}. This is illustrated by the following examples.

\begin{example} 
	Consider the RM$(1,3)$ code and RM$(2,4)$ code. The RM matrix of size $8\times 6$ constructed based on the RM$(1,3)$ code is presented in~\eqref{RMmatrix31}, and the RM matrix of size $16\times24$ constructed based on the RM$(2,4)$ code is shown in~\eqref{RMmatrix42}.
\end{example}

\begin{equation}\label{RMmatrix31}
	\mathbf{G}_{8\times 6}=\begin{bNiceArray}{ccc:ccc:ccc}
		1 & & 0 & 1 & & 0  & 1 & & 0 \\
		1 & & 0 & 1 & & 0  & 0 & &  1 \\
		1 & & 0 & 0 & & 1  & 1 & &  0 \\
		1 & & 0 & 0 & & 1  & 0 & &  1 \\
		0 & & 1 & 1 & & 0  & 1 & &  0 \\
		0 & & 1 & 1 & & 0  & 0 & &  1 \\
		0 & & 1 & 0 & & 1  & 1 & &  0 \\
		0 & & 1 & 0 & & 1  & 0 & &  1 \\
		\CodeAfter
		\OverBrace[shorten,
		yshift=3pt]{1-1}{2-3}{\{v_1\}}
		\OverBrace[shorten,
		yshift=3pt]{1-4}{2-6}{\{v_2\}}
		\OverBrace[shorten,
		yshift=3pt]{1-7}{2-9}{\{v_3\}}
	\end{bNiceArray}. 
\end{equation}
\begin{equation}\label{RMmatrix42}
	\mathbf{G}_{16 \times 24}=\begin{bNiceArray}{cccc:cccc:cccc:cccc:cccc:cccc}
		1 & 0 & 0 & 0 & 1 & 0 & 0 & 0 &  1 & 0 & 0 & 0 & 1 & 0 & 0 & 0 & 1 & 0 & 0 & 0 & 1 & 0 & 0 & 0\\
		1 & 0 & 0 & 0 & 1 & 0 & 0 & 0 &  0 & 1 & 0 & 0 & 1 & 0 & 0 & 0 & 0 & 1 & 0 & 0 & 0 & 1 & 0 & 0\\
		1 & 0 & 0 & 0 & 0 & 1 & 0 & 0 &  1 & 0 & 0 & 0 & 0 & 1 & 0 & 0 & 1 & 0 & 0 & 0 & 0 & 0 & 1 & 0\\
		1 & 0 & 0 & 0 & 0 & 1 & 0 & 0 &  0 & 1 & 0 & 0 & 0 & 1 & 0 & 0 & 0 & 1 & 0 & 0 & 0 & 0 & 0 & 1\\
		0 & 1 & 0 & 0 & 1 & 0 & 0 & 0 &  1 & 0 & 0 & 0 & 0 & 0 & 1 & 0 & 0 & 0 & 1 & 0 & 1 & 0 & 0 & 0\\
		0 & 1 & 0 & 0 & 1 & 0 & 0 & 0 &  0 & 1 & 0 & 0 & 0 & 0 & 1 & 0 & 0 & 0 & 0 & 1 & 0 & 1 & 0 & 0\\
		0 & 1 & 0 & 0 & 0 & 1 & 0 & 0 &  1 & 0 & 0 & 0 & 0 & 0 & 0 & 1 & 0 & 0 & 1 & 0 & 0 & 0 & 1 & 0\\
		0 & 1 & 0 & 0 & 0 & 1 & 0 & 0 &  0 & 1 & 0 & 0 & 0 & 0 & 0 & 1 & 0 & 0 & 0 & 1 & 0 & 0 & 0 & 1\\
		0 & 0 & 1 & 0 & 0 & 0 & 1 & 0 &  0 & 0 & 1 & 0 & 1 & 0 & 0 & 0 & 1 & 0 & 0 & 0 & 1 & 0 & 0 & 0\\
		0 & 0 & 1 & 0 & 0 & 0 & 1 & 0 &  0 & 0 & 0 & 1 & 1 & 0 & 0 & 0 & 0 & 1 & 0 & 0 & 0 & 1 & 0 & 0\\
		0 & 0 & 1 & 0 & 0 & 0 & 0 & 1 &  0 & 0 & 1 & 0 & 0 & 1 & 0 & 0 & 1 & 0 & 0 & 0 & 0 & 0 & 1 & 0\\
		0 & 0 & 1 & 0 & 0 & 0 & 0 & 1 &  0 & 0 & 0 & 1 & 0 & 1 & 0 & 0 & 0 & 1 & 0 & 0 & 0 & 0 & 0 & 1\\
		0 & 0 & 0 & 1 & 0 & 0 & 1 & 0 &  0 & 0 & 1 & 0 & 0 & 0 & 1 & 0 & 0 & 0 & 1 & 0 & 1 & 0 & 0 & 0\\
		0 & 0 & 0 & 1 & 0 & 0 & 1 & 0 &  0 & 0 & 0 & 1 & 0 & 0 & 1 & 0 & 0 & 0 & 0 & 1 & 0 & 1 & 0 & 0\\
		0 & 0 & 0 & 1 & 0 & 0 & 0 & 1 &  0 & 0 & 1 & 0 & 0 & 0 & 0 & 1 & 0 & 0 & 1 & 0 & 0 & 0 & 1 & 0 \\
		0 & 0 & 0 & 1 & 0 & 0 & 0 & 1 &  0 & 0 & 0 & 1 & 0 & 0 & 0 & 1 & 0 & 0 & 0 & 1 & 0 & 0 & 0 & 1\\
		\CodeAfter
		\OverBrace[shorten,
		yshift=3pt]{1-1}{2-4}{\{v_1, v_2\}}
		\OverBrace[shorten,
		yshift=3pt]{1-5}{2-8}{\{v_1, v_3\}}
		\OverBrace[shorten,
		yshift=3pt]{1-9}{2-12}{\{v_1, v_4\}}
		\OverBrace[shorten,
		yshift=3pt]{1-13}{2-16}{\{v_2, v_3\}}
		\OverBrace[shorten,
		yshift=3pt]{1-17}{2-20}{\{v_2, v_4\}}
		\OverBrace[shorten,
		yshift=3pt]{1-21}{2-24}{\{v_3, v_4\}}
	\end{bNiceArray}.
\end{equation}

Now define $\mathscr{C}=\{ \bm u^n \in \mathbb{F}_2^n: \bm u^n \mathbf{G} = \bm 0^{2^r \binom{m}{r}}\}$, which is the RM$(m-r-1,m)$ code with $\mathbf{G}^T$ as the parity-check matrix. Different from the Gallager's LDPC codes, the parity-check matrix $\mathbf{G}^T$ has an increasing row weight $2^{m-r}$ as $m$ increases. Since the columns of $\mathbf{G}$ are linearly dependent, we may randomly remove some redundant columns from $\mathbf{G}$. The resulting matrix is still denoted by 
$\mathbf{G} = \left[\mathbf{G}_1, \mathbf{G}_2, \cdots, \mathbf{G}_t\right]$ with $\mathbf{G}_i$ of size $n\times \ell_i$ and $\sum\limits_{i=1}^t\ell_i = \ell$ such that $\mathbf{G}$ is of full rank. To be specific, we may take $\ell_1 = 2^r$ and $\ell_i \geq \ell_{i-1}$ for $i > 1$.

\subsection{Main Results}
In the following theorems, by an RM matrix $\mathbf{G}$, we mean the matrix $\mathbf{G}$ whose columns are linearly independent codewords selected from the RM$(r, m)$ code. 

\begin{theorem}\label{SourceTheorem}
	Consider the BIOS 1 channel as a totally erased channel and the BIOS 2 channel as a noiseless channel. Let $R > H(U)$ be fixed. For any $\varepsilon >0$, one can always find a sufficiently large integer $m_0$ such that, for all $m \geq m_0$, $n = 2^m$ and $r < m$ be the minimum integer with $\sum\limits_{i=0}^{r}\binom{m}{i} \geq n R$, the BER of the linear code specified by an RM matrix $\mathbf{G}$ of size $n \times \ell$ with $\ell = \lceil nR\rceil$ is upper bounded by $\varepsilon$ under maximum likelihood decoding.
\end{theorem}

\begin{theorem}\label{RMChannelTheorem}
	Consider the BIOS 1 channel as a noisy channel and the BIOS 2 channel as a noiseless channel. Let $R > H(U|V)$ be fixed. For any $\varepsilon >0$, one can always find a sufficiently large integer $m_0$ such that, for all $m \geq m_0$, $n = 2^m$  and $r < m$ be the minimum integer with $ \sum\limits_{i=0}^{r}\binom{m}{i}   \geq nR $, the BER of the linear code specified by the parity-check (RM) matrix $\mathbf{H} = \mathbf{G}^T$ of size $\ell\times n$ with $\ell = \lceil nR \rceil$ is upper bounded by $\varepsilon$ under maximum likelihood decoding.
\end{theorem}

\begin{remark}
	Theorems~\ref{SourceTheorem} and~\ref{RMChannelTheorem} definitely hold if we choose $\ell = \sum\limits_{i=0}^{r}\binom{m}{i}$ instead of $\ell = \lceil n R \rceil$, in which case the codes are exactly the same as the RM codes but defined by the parity-check matrices. Notice that the channel codes defined by the parity-check matrices in Theorem~2 have a code rate $(n-\lceil nR\rceil)/n$, which achieves the capacity $C = 1-H(U|V)$. Also notice that the parity-check code defined by the RM matrix $\mathbf{G}$ of size $n\times \ell$ associated with the RM$(r,m)$ code is the RM$(m-r-1, m)$ code~\cite[Theorem 4, Ch. 13]{macwilliams1977theory}. 
\end{remark}

\section{Proof of Coding Theorems}
\subsection{Proof of Theorem 1}

For a fixed positive real number $R < 1$, define for each positive integer $m$ that $n = 2^m$ and $\ell = \lceil nR\rceil$ as the minimum integer not less than $nR$. 
Let $\mathbf{G}_{n\times \ell}$ be a full-rank binary RM matrix of size $n \times \ell$.


Let $\bm u^n \in \mathbb{F}_2^n$ be a source sequence, and $\bm x^\ell = \bm u^n \mathbf{G}$ is the parity-check vector, also known as syndrome. Upon receiving $\bm x^\ell$, the ML decoding finds $\hat{\bm u}^n$ with $\hat{\bm u}^n \mathbf{G} = \bm x^\ell$ such that $P\left(\hat{\bm u}^n\right) > P\left(\bm w^n\right)$ for all $\bm w^n \neq \hat{\bm u}^n$ and $\bm w^n \mathbf{G} = \bm x^\ell$. If no such solution exists~(in the case of ties), the ML decoding simply reports an error. To analyze the probability of decoding error, we consider an RM code ensemble by replacing the RM matrix $\mathbf{G}$ by $\mathbf{\Pi}\mathbf{G}$, where $\mathbf{\Pi}$ is drawn uniformly at random from all permutation matrices of order $n$. Defining $$
E(\bm u^n) =\left\{\bm w^n| \bm w^n \neq \bm u^n, p(\bm w^n) \geq p(\bm u^n), (\bm w^n-\bm u^n)\mathbf{G} = \bm 0\right\},$$ a decoding error occurs if $E(\bm u^n)$ is non-empty, i.e., $|E(\bm u^n)| \geq 1$. The BLER, denoted by $\textrm{Pr}\{E\}$, keeps unchanged after introducing $\mathbf{\Pi}$ and is then upper bounded by, for any $0\leq \gamma \leq 1$,
\begin{align}
	\textrm{Pr}\{E\} & = \sum_{\bm u^n \in \mathbb{F}_2^n} P(\bm u^n) \cdot \mathbb{E}\left[\mathbb{I}(|E(\bm u^n)| \geq 1)\right] \nonumber\\
	&\leq \sum_{\bm u^n \in \mathbb{F}_2^n} P(\bm u^n) \bigg( \sum_{\bm w^n\in \mathbb{F}_2^n} \left(\frac{P(\bm w^n)}{P(\bm u^n)}\right)^{1/(1+\gamma)} \nonumber\\
	&\quad \cdot \mathbb{E}\left[\mathbb{I}\left((\bm w^n-\bm u^n)\mathbf{\Pi}\mathbf{G}=\bm 0  \right)\right]\bigg)^\gamma,
\end{align}
where we use $\mathbb{I}(\cdot)$ to represent the indicator function. This bounding technique is similar to that used by~\cite[Chapter 5.6 and Exercise 5.16]{Gallager1968}, where the indicator function is introduced to limit the number of terms in the inner summation. Also notice that, to prove the coding theorems in terms of BER, only those $\bm w^n$ with $W_H(\bm w^n-\bm u^n) \geq n\epsilon$ need to be included in the inner summation. This simplifies much the argument. 

Defining the error exponent
\begin{align}
	E(\gamma, n, \ell)\! &= \!-\frac{1}{n}\log \bigg\{\sum_{\bm u^n} P(\bm u^n)^{1/(1+\gamma)}\! \bigg(\!\sum_{\bm w^n} \!\!P(\bm w^n)^{1/(1+\gamma)} \nonumber\\
	&\quad \cdot\mathbb{E}\left[\mathbb{I}\left((\bm w^n-\bm u^n)\mathbf{\Pi}\mathbf{G}=\bm 0\right)\right] \bigg)^\gamma\bigg\},
\end{align}
where $\log$ and $\exp$ are both taking base $2$,
we have 
\begin{align}
	\textrm{Pr}\{E\} \leq \exp\{-nE(\gamma, n, 
	\ell)\}.
\end{align}

Noticing that $E(0, n, \ell) = 0$ and
\begin{align}
	\frac{\partial 	E(\gamma, n, \ell)}{\partial \gamma} \bigg|_{\gamma = 0}\!\!\!\!\!&= \!-H(U) \!- \!\frac{1}{n}\!\sum\limits_{\bm u^n}\! P(\bm u^n) \log\!\bigg\{\sum_{\bm w^n} P(\bm w^n) \nonumber \\
	&\quad \cdot \mathbb{E}\left[\mathbb{I}\left((\bm w^n-\bm u^n)\mathbf{\Pi}\mathbf{G}=\bm 0\right)\right] \bigg\},
\end{align}
we see that the key is to prove that $\frac{\partial 	E(\gamma, n, \ell)}{\partial \gamma} \big|_{\gamma = 0} > 0$. The difficulty lies in the fact that $\mathbb{I}\left((\bm w^n-\bm u^n)\mathbf{\Pi}\mathbf{G}=0\right)$ is a product of $t$ binary random vectors $\mathbb{I}\left((\bm w^n-\bm u^n)\mathbf{\Pi}\mathbf{G}_i=\bm 0\right)$ of length $\ell_i$ that may be dependent, where $\mathbf{G}_i$ is the $i$-th sub-matrix of $\mathbf{G}$ as specified in Sec.~II.~B.
 Nevertheless, we have the following lemma that ensures,  for any $\bm w^n$ and $\bm u^n$ such that $n\epsilon \leq W_H(\bm w^n-\bm u^n) \leq n(1-\epsilon)$ with a given $0 < \epsilon < 1$, 
\begin{equation}
	\mathbb{E}[\mathbb{I}\left((\bm w^n-\bm u^n)\mathbf{\Pi}\mathbf{G}_i=\bm 0\right)] \leq\left(\frac{1}{2}+\delta\right)^{\ell_i}
\end{equation}
 as $n$ goes to infinity. Then we make the following assumption.
\begin{assumption}
The $t$ random binary vectors $\mathbb{I}\left((\bm w^n-\bm u^n)\mathbf{\Pi}\mathbf{G}_i=\bm 0\right)$, $1\leq i \leq t$, are independent.
\end{assumption}	 
 Under this assumption, which holds if all sub-matrices are randomly interleaved to each other, we have 
\begin{equation}
	\frac{\partial 	E(\gamma, n, \ell)}{\partial \gamma} \big|_{\gamma = 0} \geq \frac{\ell}{n} - H(U)  > 0
\end{equation}
for $R > H(U)$. This completes the proof.

\begin{lemma}
	Let $\bm w^n \in \mathbb{F}_2^n$ be a vector of Hamming weight $w$ and $\delta$ be an arbitrarily small positive number. We have 
	\begin{enumerate}
		\item $w \in \{0,n\}$, $\mathbb{I}\left(\bm w^n\mathbf{\Pi}\mathbf{G}=\bm 0\right) =1$ since both $\bm w^n = \bm 0$ and $\bm w^n = \bm 1$ are codewords in the RM$(m-r-1, m)$ code.
		\item $w \in \{i:1 < i < 2^{r+1}\}\cup\{i:n-2^{r+1} < i < n \}\cup\{i: i \textrm{ mod } 2 = 1\}$, $\mathbb{I}\left(\bm w^n\mathbf{\Pi}\mathbf{G}=\bm 0\right) =0$ since all codewords must have even weights at least $2^{r+1}$. 
		\item $ w \in\{i: n\epsilon \leq i\leq n(1-\epsilon)\}$  with $\epsilon$ being fixed as any arbitrarily small positive real number, $\mathbb{E}[\mathbb{I}\left(\bm w^n\mathbf{\Pi}\mathbf{G}=\bm 0\right)] \leq \left(\frac{1}{2}+\delta\right)^{\ell}$.
	\end{enumerate}
\end{lemma}
\begin{proof}
	To derive the upper bound, we take as an example the event $\mathbb{I}\left(\bm w^n\mathbf{\Pi}\mathbf{G}_1=\bm 0\right)=1$, where $\mathbf{G}_1$ is the first sub-matrix of $\mathbf{G}$. 
%
	We consider a model in which we throw an even number
	$w \geq 2^{r+1}$ balls into $n = 2^m$ bins uniformly at random. These bins are divided into $\ell_1 = 2^r$ sub-bins, each having volume of $2^{m-r}$. Then $\mathbb{I}\left(\bm w^n\mathbf{\Pi}\mathbf{G}_1=\bm 0\right)=1$ is equivalent to that each of these sub-bins has an even number of balls.
	
	These process can be done recursively by $r$ steps. The first step is to throw the $w$ balls into two sub-bins of volume $2^{m-1}$. Then the balls in each sub-bins are further distributed into two sub-bins of volume $2^{m-2}$. This process stops after the $r$-th step.
	
	To satisfy $\mathbb{I}\left(\bm w^n\mathbf{\Pi}\mathbf{G}_1=\bm 0\right)=1$, each step must distribute the balls into two equal sub-bins with two even numbers. It can be proved that, for sufficiently large $n$ and $n\epsilon \leq w \leq n(1-\epsilon)$, we have 
	\begin{align}
		\mathbb{E}\left[\mathbb{I}\left(\bm w^n\mathbf{\Pi}\mathbf{G}_1=\bm 0\right)\right] \leq \left(\frac{1}{2}+\delta\right)^{\ell_1 }.
	\end{align}
	
	Therefore, under Assumption 1, we have 
	\begin{align}
			\mathbb{E}\left[\mathbb{I}\left(\bm w^n\mathbf{\Pi}\mathbf{G}=\bm 0\right)\right]\leq \left(\frac{1}{2}+\delta\right)^{\ell}.
	\end{align}
\end{proof}
\subsection{Proof of Theorem 2}
It is well known that compressing a Bernoulli source with success probability $\theta$ can be viewed as protecting the all-zero codeword over a binary symmetric channel~(BSC) with a cross error probability $\theta$.  Now we show that the channel coding for BIOS channels can be transformed into the source coding, completing the proof of Theorem 2.

Suppose that $\bm U^n$ with $\bm U^n  \mathbf{G} = \bm 0$ is transmitted over BIOS1, resulting in $\bm V^n$. Then the log-likelihood ratio~(LLR) vector $\boldsymbol{r}^n$ is calculated as
\begin{equation}\label{eq1}
	r_i = \ln{\frac{P_{V|U}(V_i|U_i=0)}{P_{V|U}(V_i|U_i = 1)}},\ 0 \leq i < n.
\end{equation}
Given the LLR vector $\boldsymbol{r}^n$, the hard-decision vector $\boldsymbol{Z}^n \in \mathbb{F}_2^{n}$ is calculated as
\begin{equation}\label{harddecison}
	Z_i=
	\begin{cases} 
		0, & \mbox{if }r_i > 0\\
		0 \text{ or } 1 \text{ with probability } 1/2, &\mbox{if } r_i = 0\\
		1, & \mbox{if }r_i<0\\
	\end{cases}
\end{equation}
for $0\leq i < n$. 

Then the channel can be transformed into a possibly time-varying BSC channel, 
\begin{equation}
	\bm Z^n = \bm U^n + \bm E^n,
\end{equation} 
where $\bm E^n$ is the error pattern with 
\begin{equation}
	\theta_i = P(E_i = 1) = \frac{1}{1+\exp(|r_i|)}
\end{equation}
for $0\leq i < n$.

Since $\bm Z^n \mathbf{G} = (\bm U^n + \bm E^n)\mathbf{G} = \bm E^n\mathbf{G}$, the decoding problem is equivalent to finding $\bm E^n$ from its compressed version $\bm X^{\ell} = \bm Z^n\mathbf{G}$~(called the syndrome). The slight difference from the source coding is that the error pattern $\bm E^n$ may have time-varying $\theta_i$ depending on $V_i$. This, however, does not preclude we define $H(E) = \lim\limits_{n\rightarrow \infty} \frac{1}{n} \sum\limits_{t=0}^{n-1} H(E_t)=H(U|V)$. Hence, we can derive similarly an error exponent under Assumption 1,
\begin{equation}
\begin{aligned}
		E(\gamma, n, \ell) &=\! -\frac{1}{n}\log\! \bigg\{\sum_{\bm e^n} \!P(\bm e^n)^{1/(1+\gamma)} \!\bigg( \sum_{\bm w^n\in \mathbb{F}_2^n}\! P(\bm w^n)^{1/(1+\gamma)} \\
	&\quad \cdot \mathbb{E}\left[\mathbb{I}\bigg((\bm w^n-\bm e^n)\mathbf{\Pi}\mathbf{G}=\bm 0\bigg)\right] \bigg)^\gamma\bigg\},
\end{aligned}
\end{equation}
and
\begin{equation}
	\textrm{Pr}\{E\} \leq \exp\{-nE(\gamma, n, \ell)\}.
\end{equation}

Therefore, we have $E(0, n, \ell) = 0$ and
\begin{align}
	&\frac{\partial 	E(\gamma, n, \ell)}{\partial \gamma} \bigg|_{\gamma = 0}\nonumber\\
	&= -\!H(U|V)- \frac{1}{n}\sum\limits_{\bm e^n} P(\bm e^n) \nonumber\\
		&\quad \cdot \log\bigg\{\sum_{\bm w^n} P(\bm w^n) 
\mathbb{E}\left[\mathbb{I}\left((\bm w^n-\bm e^n)\mathbf{\Pi}\mathbf{G}=\bm 0\right)\right] \bigg\}\nonumber\\
	&\geq \frac{\ell}{n} - H(U|V)  > 0
\end{align}
for $R > H(U|V)$. This holds for all typical receiving sequences $\bm V^n$ as defined in~\cite{cover1999elements}, completing the proof.

\section{Conclusions}
We have proved the coding theorems for RM codes under an independence assumption, in terms of BER under ML decoding. The proof could be extended to the coding theorems in terms of BLER if new techniques~\cite{wang2024coding,wang2024Systematic} are introduced.

\section*{Acknowledgment}
The author would like to thank Mr. Jifan Liang, Mr. Yixin Wang and Ms. Xiangping Zheng for their helpful discussions, who helped programming to verify some intuitive ideas. This paper was typed by Ms. Xiangping Zheng.

\bibliographystyle{IEEEtran}
\bibliography{ref}

\end{document}